\documentclass[a4paper,10pt]{article}
\usepackage{authblk}

\usepackage{amsmath,amsthm,amssymb}
\usepackage{mathtools}

\usepackage{enumerate}
\usepackage{braket}

\usepackage{hyperref}
\usepackage[numbers]{natbib}
\usepackage[nameinlink]{cleveref}
\usepackage{autonum}

\theoremstyle{plain}
\newtheorem{theorem}{Theorem}
\newtheorem{proposition}{Proposition}
\newtheorem{lemma}{Lemma}
\newtheorem{definition}{Definition}
\newtheorem{corollary}{Corollary}
\newtheorem{setting}{Setting}

\crefname{theorem}{Theorem}{Theorems}
\crefname{proposition}{Proposition}{Propositions}
\crefname{lemma}{Lemma}{Lemmas}
\crefname{definition}{Definition}{Definitions}
\crefname{corollary}{Corollary}{Corollaries}
\crefname{setting}{Setting}{Settings}
\crefname{table}{Table}{Tables}
\crefname{equation}{Eq.}{Eqs.} 
\crefname{section}{Section}{Sections}

\DeclareMathOperator*{\argmin}{argmin}

\title{Quantum Communication Complexity of Regularized Linear Regression Protocols\footnote{Accepted for publication in IEEE Transactions on Quantum Engineering.

DOI: 10.1109/TQE.2026.3687237}}
\author{
    Sayaki Matsushita\thanks{sayaki-m@nagoya-u.jp}
}
\affil{Department of Mathematical Informatics, Graduate School of Informatics, Nagoya University, Japan}
\date{}

\begin{document}

\maketitle

\begin{abstract}
Linear regression is fundamental to statistical analysis and machine learning, but its application to large-scale datasets necessitates distributed computing. The problem also arises in quantum computing, where handling extensive data requires distributed approaches. This paper investigates distributed linear regression in the quantum coordinator model.

Building upon the distributed quantum least squares protocol developed by Montanaro and Shao \cite{MS}, I propose improved and extended quantum protocols for solving both ordinary (unregularized) and L2-regularized (Tikhonov) least squares problems. For ordinary least squares methods, my protocol reduces the quantum communication complexity compared to the previous protocol. In particular, this yields a quadratic improvement in the number of digits of precision required for the generated quantum states. This improvement is achieved by incorporating advanced techniques such as branch marking and branch-marked gapped phase estimation developed by Low and Su \cite{LS}. 

Furthermore, I establish a setting for the $\ell_2$-regularized least squares problem specifically in the quantum coordinator model and derive its quantum communication complexity. I analyze the effect of regularization parameters on the quantum communication complexity. 
\end{abstract}

\section{Introduction}\label{section:introduction}
Linear regression, which models the linear relationship between variables, is widely used in statistical analysis and machine learning. The least squares method is a foundational approach in linear regression.

I now introduce the least squares problem. For $m$ data pairs $(\mathbf{a}_1, b_1), \cdots, (\mathbf{a}_m, b_m)$, where each $\mathbf{a}_i\in\mathbb{R}^n$ represents $n$ variables and $ b_i\in \mathbb{R}$ is the corresponding response variable, the least squares method aims to find a vector $\mathbf{x}\in \mathbb{R}^n$ that minimizes the sum of squared differences:
\begin{equation}
    \sum_{i=1}^m (\mathbf{a}_i^T\mathbf{x} - b_i)^2.
\end{equation}
To express this minimization in a more compact matrix form, I introduce the matrix $A\in\mathbb{R}^{m\times n}$ and a vector $\mathbf{b}\in\mathbb{R}^m$. The matrix $A$ is formed by stacking the transposed feature vectors such that its $i$-th row is $\mathbf{a}_i^T$, meaning $A = (\mathbf{a}_1, \cdots, \mathbf{a}_m)^T$. The vector $\mathbf{b}$ is a column vector $\mathbf{b} = (b_1, \cdots, b_m)^T$. With these definitions, the minimization term can be rewritten as:
\begin{equation}\label{eq:ls_obj}
    \|A\mathbf{x}-\mathbf{b}\|^2=\sum_{i=1}^m (\mathbf{a}_i^T\mathbf{x} - b_i)^2.
\end{equation}

The computational infeasibility of processing large-scale datasets on a single machine necessitates distributed approaches, where data is distributed across multiple nodes and processed through internode communication.

Vempala, Wang, and Woodruff \cite{VWW} investigated least squares regression in distributed settings, including coordinator model. The coordinator model involves multiple parties and a single referee, where each party communicates solely with the referee, and direct communication between parties is not permitted.

\begin{setting}[Coordinator Model for Linear Regression, \cite{VWW}]\label{setting:cls}
Suppose that there are $r$ parties $P_0, \cdots, P_{r-1}$, and each party $P_i$ has a matrix $A_i\in\mathbb{R}^{m_i\times n}$ and a vector $\mathbf{b}_i\in \mathbb{R}^{m_i}$. Let a matrix $A\in\mathbb{R}^{m\times n}$ and a vector $\mathbf{b}\in\mathbb{R}^m$ be
\begin{equation}
    A:=\begin{pmatrix}
    A_0\\
    \vdots\\
    A_{r-1}
\end{pmatrix},\  \mathbf{b}:=\begin{pmatrix}
    \mathbf{b}_0\\
    \vdots\\
    \mathbf{b}_{r-1}
\end{pmatrix},
\end{equation}
where $m=\sum_{i=0}^{r-1}m_i$.

Suppose there is a single referee, and each party can make classical communication with the referee. 
Then the goal is for the referee to obtain an $\varepsilon$-approximation of the solution vector:
\begin{equation}\label{eq:x_opt}
\mathbf{x}_\mathrm{opt}=\argmin_\mathbf{x} \|A\mathbf{x}-\mathbf{b}\|.
\end{equation}
\end{setting}

According to \cite{VWW}, the communication complexity of the least squares regression in the coordinator model is $\Tilde{\Omega}(rn)$.

Quantum algorithms are also emerging to solve the least squares problem. A seminal contribution is the HHL algorithm, introduced by Harrow, Hassidim, and Lloyd \cite{HHL}, which was originally designed to solve linear systems ($A\mathbf{x}=\mathbf{b}$) by preparing a quantum state proportional to the solution vector $\mathbf{x}$. This algorithm can be adapted to find the least squares solution $\mathbf{x}_\mathrm{opt}:= \argmin_{\mathbf{x}}\|A\mathbf{x}-\mathbf{b}\|$ by efficiently generating a quantum state $\ket{\mathbf{x}_\mathrm{opt}}=\mathbf{x}_\mathrm{opt}/\|\mathbf{x}_\mathrm{opt}\|$. Subsequent improvements for generating $\ket{\mathbf{x}_\mathrm{opt}}$ have also been made.
Ambainis \cite{ambainis} developed variable-time amplitude amplification (VTAA) and improved the dependency of the conditional number. Childs, Kothari, and Somma \cite{CKS} proposed methods for approximating $A^{-1}$ using either a Fourier series or Chebyshev polynomials, and  achieved an exponential improvement in precision dependence. Chakraborty, Gily\'en, and Jeffery \cite{CGJ} incorporated block-encoding technique and enabled quantum least squares regression for general matrices beyond sparse matrices. Recently, Low and Su \cite{LS} proposed advanced techniques such as Branch Marking and branch-marked Gapped Phase Estimation, further enhancing the efficiency of least squares regression.

In the context of near-term hardware, several studies have experimentally implemented the HHL algorithm using a limited number of qubits \cite{DSDRBP}.
Additionally, approaches for performing least squares regression using quantum annealers have been proposed in \cite{VMAKD} and \cite{DP}.

Similar to classical methods, quantum computation also explores distributed data settings.

Montanaro and Shao \cite{MS} introduced a quantum setting as a quantum analog of the coordinator model \cite{VWW}, establishing a specific distributed quantum computing setting for linear regression. 

In this setting, quantum communication means that parties and the referee exchange qubits.
\begin{setting}[Quantum Coordinator Model for Linear Regression, \cite{MS}]\label{setting:ls}
Suppose that there are $r$ parties $P_0, \cdots, P_{r-1}$, and each party $P_i$ has a matrix $A_i\in\mathbb{R}^{m_i\times n}$ and a vector $\mathbf{b}_i\in \mathbb{R}^{m_i}$. Let a matrix $A\in\mathbb{R}^{m\times n}$ and a vector $\mathbf{b}\in\mathbb{R}^m$ be
\begin{equation}\label{eq:qcm_of_A_b}
    A:=\begin{pmatrix}
    A_0\\
    \vdots\\
    A_{r-1}
\end{pmatrix},\  \mathbf{b}:=\begin{pmatrix}
    \mathbf{b}_0\\
    \vdots\\
    \mathbf{b}_{r-1}
\end{pmatrix},
\end{equation}
where $m=\sum_{i=0}^{r-1}m_i$.

Suppose there is a single referee, and each party can make quantum communication with the referee. Suppose that the referee has prior knowledge of $\delta$, which represents the lower bound of the minimum non-zero singular value of $A$.
Then the goal is for the referee to generate an $\varepsilon$-approximation of the quantum state $\ket{\mathbf{x}_\mathrm{opt}}:=\mathbf{x}_\mathrm{opt}/\|\mathbf{x}_\mathrm{opt}\|$, where $\mathbf{x}_\mathrm{opt}$ is defined as \cref{eq:x_opt}.
\end{setting}

In most real-world scenarios, the referee generally cannot know the value of $\delta$ in advance. 
Practically, I treat $\delta$ as a cutoff threshold and the protocol addresses this by ignoring singular values less than $\delta$. This means the protocol treats the matrix $A$ as $A'=\sum_{k\in \{k|\sigma_k\geq \delta\}}\sigma_k\ket{x_k}\bra{y_k}$, where the singular value decomposition of $A$ is $A=\sum_k\sigma_k\ket{x_k}\bra{y_k}$.

Montanaro and Shao \cite{MS} developed a quantum protocol for \cref{setting:ls} and calculated its quantum communication complexity. However, their protocol is based on Childs, Kothari, and Somma's gapped phase estimation \cite{CKS}, which leaves room for improving the precision dependence. Specifically, their protocol uses an approximation of $\exp(i\pi A)$ to get eigenvalue information, and this approximation leads to suboptimal precision dependence.

In this work, I propose a protocol that reduces quantum communication complexity compared to the approach by Montanaro and Shao \cite{MS}. This advancement is achieved by incorporating Branch Marking and branch-marked Gapped Phase Estimation developed in \cite{LS}. I adapt their underlying framework, Quantum Signal Processing (QSP), so that it can be executed in the quantum coordinator model. Furthermore, the algorithm of \cite{LS} assumes prior access to a constant-factor estimate of $\|A^{+}\mathbf{b}\|$, while it is unavailable in my distributed setting. To address this, my protocol modifies the execution order of subalgorithms from \cite{LS}.

Notably, my protocol achieves a quadratic improvement in the number of digits of precision (i.e., $\log(1/\varepsilon)$) of the generated quantum state relative to prior methods. To illustrate this improvement practically, I set the precision to $\varepsilon=10^{-6}$. The protocol of \cite{MS} requires a precision-dependent overhead $\log^2(1/\varepsilon)\approx 191$ , whereas my protocol requires $\log(1/\varepsilon)\approx 13.8$.

The least squares method I have discussed thus far, aiming to minimize $\|A\mathbf{x}-\mathbf{b}\|^2$, is also known as Ordinary Least Squares. While the ordinary least squares method is foundational, its practical application often encounters challenges such as overfitting and multicollinearity.  Regularization techniques, which introduce additional penalties to the model-fitting process, are used to mitigate these issues. Ridge regression is an example of such a regularization technique \cite{MPVtext, JWHTtext}. It introduces an $\ell_2$-norm penalty to the objective function of the ordinary least squares:
\begin{equation}\label{eq:ridge_obj}
    \mathcal{L}_\mathrm{ridge}(\mathbf{x})=\|A\mathbf{x}-\mathbf{b}\|^2+\lambda\|\mathbf{x}\|^2,
\end{equation}
for a hyperparameter $\lambda>0$. 

General $\ell_2$-regularization (or Tikhonov regularization) is a generalization of the ridge regression. This method incorporates a penalty term defined by a full-rank penalty matrix $L\in\mathbb{R}^{n\times n}$, leading to the following objective function:
\begin{equation}\label{eq:l2_obj}
    \mathcal{L}_\mathrm{l2}(\mathbf{x})=\|A\mathbf{x}-\mathbf{b}\|^2+\lambda\|L\mathbf{x}\|^2.
\end{equation}
Then the goal is to calculate $\mathbf{x}_\mathrm{l2}=\argmin_\mathbf{x}\mathcal{L}_\mathrm{l2}(\mathbf{x})$.

I define a specific setting for performing $\ell_2$-regularized least squares within the quantum coordinator model. This setting extends \cref{setting:ls} by incorporating the regularization parameters directly into the referee's knowledge:
\begin{setting}[Quantum Coordinator Model for $\ell_2$-regularized Linear Regression]\label{setting:l2reg}
Suppose there are $r$ parties $P_0, \cdots, P_{r-1}$, and each party $P_i$ has a matrix $A_i\in\mathbb{R}^{m_i\times n}$ and a vector $\mathbf{b}_i\in \mathbb{R}^{m_i}$. Let
\begin{equation}
    A=\begin{pmatrix}
    A_0\\
    \vdots\\
    A_{r-1}
\end{pmatrix},\  \mathbf{b}=\begin{pmatrix}
    \mathbf{b}_0\\
    \vdots\\
    \mathbf{b}_{r-1}
\end{pmatrix}.
\end{equation}
Suppose there is a single referee with a hyperparameter $\lambda>0$ and a penalty matrix $L\in \mathbb{R}^{n\times n}$. Suppose each party can make quantum communication with the referee.
Let
\begin{equation}\label{eq:x_l2}
    \mathbf{x}_\mathrm{l2}:=\argmin_{\mathbf{x}}(\|A\mathbf{x}-\mathbf{b}\|^2+\lambda\|L\mathbf{x}\|^2).
\end{equation}
Then the goal is for the referee to generate an $\varepsilon$-approximation of the quantum state $\ket{\mathbf{x}_\mathrm{l2}}:=\mathbf{x}_\mathrm{l2}/\|\mathbf{x}_\mathrm{l2}\|$.
\end{setting}

I develop a quantum protocol for \cref{setting:l2reg} and derive its quantum communication complexity. 
While $\ell_2$-regularization can be incorporated straightforwardly in the single-machine setting, its formulation in a distributed quantum setting requires careful consideration. Furthermore, the parameters $L$ and $\lambda$ nontrivially affect the quantum communication complexity because the augmented matrix constructed from $A$ generally has a diffirent condition number than that of $A$ itself.

\subsection{Related Works}
Quantum algorithms for ridge regression for a single quantum computer have been proposed in \cite{CMP}.
I also note that there exists research on quantum federated learning (QFL) for linear regression \cite{YGL}. While QFL also addresses distributed settings, it focuses on preventing data leakage and uses gradient descent methods to approximate the solution vector $\mathbf{x}_\mathrm{opt}$, resulting in a communication complexity that is polynomial in the dimension $n$. In contrast, my protocol aims to generate the quantum state $\ket{\mathbf{x}_\mathrm{opt}}$ with an exponentially improved communication complexity that is polylogarithmic in $n$.

This paper is organized as follows: \cref{section:preliminaries} provides an overview of the necessary prerequisites and foundational tools, including block encoding, quantum signal processing, and variable-time amplitude amplification. Based on these foundations, \cref{section:quantum_protocol} details my quantum protocol for the ordinary distributed linear regression problem of \cref{setting:ls} and analyzes its quantum communication complexity. \cref{section:l_2-regularization} extends this work by presenting my quantum protocol for the $\ell_2$-regularized linear regression problem in \cref{setting:l2reg}, and similarly derives its quantum communication complexity. \cref{section:conclusion} summarizes my contributions. \cref{section:future_work} discusses potential avenues for future research. \hyperref[section:appendix]{Appendix} analyzes the quantum communication complexity of executing the Gapped Phase Estimation (GPE) protocol in \cite{MS} and provides a brief proof for the quantum communication complexity of \cite{MS}'s protocol as presented in \cref{table:qc_ls}.

\section{Preliminaries}\label{section:preliminaries}
    To facilitate the understanding of my quantum protocol, I first review some essential concepts.
Hereafter, $\mathcal{H}_\mathsf{C}$ denotes the Hilbert space corresponding to the register $\mathsf{C}$. $X$, $Y$ and $Z$ denote Pauli matrices, $I$ denotes the identity matrix, and $H$ denotes the Hadamard matrix. Let $Op(a,b,c,d) := a I + i b Z + i c X + i d Y$ denote a linear combination of Pauli operators with scalar coefficients.
For any matrix $A$, $A^+$ denotes its Moore-Penrose pseudoinverse, and define $\Bar{A}$ by $\Bar{A}=\begin{pmatrix}
    0 & A\\
    A^\dagger & 0
\end{pmatrix}$. The $Z$-axis rotation on the Bloch sphere, $R_Z(\theta)$, is defined as $R_Z(\theta):=\exp(-i(\theta/2)Z)$.

\subsection{Block Encoding}
Only unitary operations can be directly applied to quantum states in quantum computation, so non-unitary classical matrices cannot be directly applied.  
Block Encoding provides a powerful framework for handling non-unitary matrices in quantum algorithms. The basic idea is to embed a non-unitary matrix into a larger unitary matrix, making it suitable for quantum computation.
In my protocol, Block encoding is used to handle the classical data matrix $A$ and its extended forms. 

\begin{definition}[Block-Encoding, {\cite[Definition 43]{GSLW}}]\label{definition:be}
Let $A$ be an $s$-qubit operator, $\alpha>0, \varepsilon\geq 0$  and $a\in\mathbb{N}$. An $(s+a)$-qubit unitary operator $U$ is an $(\alpha,a,\varepsilon)$-block-encoding of $A$, if
\begin{equation}\|A-\alpha\bra{0}^{\otimes a}U\ket{0}^{\otimes a}\|\leq \varepsilon.\end{equation}
\end{definition}

Suppose $U$ is an $(\alpha, a,\varepsilon)$-block-encoding of $A$ as per \cref{definition:be}. Then, the $(s+a+1)$-qubit unitary operator $\Bar{U}$ is Hermitian and an $(\alpha, a,\varepsilon)$-block-encoding of $\Bar{A}$.

\subsection{Quantum Signal Processing}
Quantum Signal Processing (QSP) offers a powerful framework for applying polynomial functions to the eigenphases of a given unitary operator. To execute quantum least squares regression, efficiently extracting information about the singular values of the data matrix $A$ is crucial. By implementing a walk operator constructed from the block-encoding of $A$, QSP can be used to access the information about the singular values of $A$. QSP is important in achieving tasks such as Branch Marking and Gapped Phase Estimation, both of which are essential components of my quantum protocols.

Not all arbitrary polynomials can be implemented through QSP.
An achievable polynomial tuple refers to a specific set of four polynomials $(f_A, f_B, f_C, f_D)$ that can be applied via QSP. In the following definition, $l$ denotes the number of required rotation gates, and $\Phi$ represents the specific sequence of rotation angles needed  to implement the desired unitary transformation corresponding to the achievable polynomial tuple.

\begin{definition}[Achievable Polynomial Tuples, \cite{LYC}]
A tuple of polynomials $(f_A, f_B, f_C, f_D)$ is called an achievable polynomial tuple if there exist $l\in\mathbb{N}$ and a sequence of angles $\Phi=(\phi_1, \cdots, \phi_l)\in\mathbb{R}^l$ such that the unitary operator
\begin{equation}\label{eq:QSP_unitary}
    U_\Phi(\theta)=R_{\phi_l}(\theta)R_{\phi_{l-1}}(\theta)\cdots R_{\phi_1}(\theta)
\end{equation}
satisfies
\begin{equation}
    U(\theta)=\begin{cases}
    Op(f_A(x),f_B(x),f_C(y),f_D(y)) & \text{if $l$ is odd}\\
    Op(f_A(x),f_B(x),xf_C(y),xf_D(y)) & \text{if $l$ is even}\\
\end{cases}
\end{equation}
where 
\begin{gather}
    R_\phi(\theta)=\exp\left(-i\frac{\theta}{2}{\sigma_\phi}\right),\quad {\sigma}_\phi=\cos(\phi)X+\sin(\phi)Y,\\
    x=\cos\frac{\theta}{2},\quad y=\sin\frac{\theta}{2}.
\end{gather}
\end{definition}

QSP constructs a desired unitary operator $U(\theta)$ by using a sequence of $R_\phi(\theta)$ rotations.

I now introduce two specific achievable polynomial tuples used in my sub-routines: Branch Marking and Gapped Phase Estimation. \cref{lemma:BM_APT} constructs an operator that is $\varepsilon$-close to $\pm iX$ for specific angle $\theta$. 

\begin{lemma}[Achievable Polynomial Tuple for Branch Marking, {\cite[Appendix D]{LS}}]\label{lemma:BM_APT}
For any $\varepsilon>0$, there exists an achievable polynomial tuple $(f_A, f_B, f_C, f_D)$ of degree at most $O(\log(1/\varepsilon))$ satisfying:
\begin{equation}
\|iX-Op(f_A(\theta),f_B(\theta),f_C(\theta),f_D(\theta))\|\leq\varepsilon
\end{equation}
for all $\theta \in [\pi/3,2\pi/3]$ and 
\begin{equation}
\|(-iX)-Op(f_A(\theta),f_B(\theta),f_C(\theta),f_D(\theta))\|\leq\varepsilon
\end{equation}
for all $\theta \in [-2\pi/3,-\pi/3]$.
\end{lemma}

\cref{lemma:GPE_APT} constructs an operator that is $\varepsilon$-close to $I$, $iX$ or $-I$ based on the angle $\theta$.

\begin{lemma}[Achievable Polynomial Tuple for Gapped Phase Estimation, {\cite[Appendix D]{LS}}]\label{lemma:GPE_APT}
For any $0<\varphi<1$, $\varepsilon>0$, and constant $\rho>1$, there exists an achievable polynomial tuple $(f_A, f_B, f_C, f_D)$ of degree at most $O((1/\varphi)\log(1/\varepsilon))$ satisfying:
\begin{equation}\|I-Op(f_A(\theta),f_B(\theta),f_C(\theta),f_D(\theta))\|\leq\varepsilon\end{equation}
for all $\theta \in (0,\arccos(\varphi)]$,
\begin{equation}\|iX-Op(f_A(\theta),f_B(\theta),f_C(\theta),f_D(\theta))\|\leq\varepsilon\end{equation}
for all $\theta \in [\arccos(\varphi/\rho),\pi-\arccos(\varphi/\rho)]$ and
\begin{equation}\|(-I)-Op(f_A(\theta),f_B(\theta),f_C(\theta),f_D(\theta))\|\leq\varepsilon\end{equation}
for all $\theta \in [\pi-\arccos(\varphi),\pi)$.
\end{lemma}

Quantum signal processing is inherently a one-qubit technique. The following lemma introduces its extension to multiple qubits.

\begin{lemma}[Quantum Signal Processing for Multiple Qubits, \cite{LC16}]\label{lemma:QSP_multi}
Let $d$ be a positive integer, and let $(f_A, f_B, f_C, f_D)$ be an achievable polynomial tuple of degree at most $d$. Given oracular access to a unitary operator $U=\sum_u\exp(i\theta_u)\ket{\phi_u}\bra{\phi_u}$, there exists a quantum circuit to realize the operator
\begin{equation}
\sum_uOp(f_A(\theta_u),f_B(\theta_u),f_C(\theta_u),f_D(\theta_u))\otimes\ket{\phi_u}\bra{\phi_u},
\end{equation}
using at most $2d$ queries to $U$ and $U^\dagger$.
\end{lemma}
Here is the construction of the quantum circuit presented in \cref{lemma:QSP_multi}, as described in \cite{LC16}. Let registers $\mathsf{P}$ and $\mathsf{Q}$ denote a single-qubit ancilla and an $n$-qubit register, respectively.
Let a unitary operator $U_\phi$ be defined as
\begin{equation}
    U_\phi:=R_Z(\phi)_\mathsf{P} H_\mathsf{P}(\ket{0}\bra{0}_\mathsf{P}\otimes I_\mathsf{Q}+\ket{1}\bra{1}_\mathsf{P}\otimes U_\mathsf{Q})H_\mathsf{P}R_Z^\dagger(\phi)_\mathsf{P}.
\end{equation}
As shown in \cite{LC16}, this operator can be expressed in the eigenbasis of $U$ as
\begin{equation}U_\phi=\bigoplus_{u}(R_\phi(\theta_u)\otimes\ket{\phi_u}\bra{\phi_u})\end{equation}
where $R_\phi(\theta)=\exp(-i(\theta/2)\sigma_\phi)$ and $\sigma_\phi=\cos(\phi)X+\sin(\phi)Y$. Assuming that $(f_A, f_B, f_C, f_D)$ is an achievable polynomial tuple, there exists a sequence of angles $\Phi\in\mathbb{R}^l$ such that the product of these operators yields
\begin{align}
&\prod_{i=1}^lU_{\phi_i}=\bigoplus_{u}(U_\Phi(\theta_u)\otimes\ket{\phi_u}\bra{\phi_u})\\
&=\sum_u(Op(f_A(\theta_u),f_B(\theta_u),f_C(\theta_u),f_D(\theta_u))\otimes\ket{\phi_u}\bra{\phi_u},
\end{align}
where $U_\Phi$ is defined as \cref{eq:QSP_unitary}.
The construction up to this point is based on \cite{LC16}.

QSP applies polynomial functions to the eigenphases of a unitary operator. However, the eigenphases of the block-encoding of $A$ do not directly correspond to the singular values of $A$. 
So, I introduce a walk operator. The eigenphases of the walk operator constructed from the block-encoding of $A$ are directly related to the singular values of $A$. The walk operator is derived from quantum walk techniques.
\begin{lemma}[Walk Operator {\cite[Corollary 19]{LS}}]\label{lemma:Q_walk}
    Let $U$ be a unitary operator and $G$ be an isometry such that $G^\dagger UG$ is Hermitian and $G^\dagger U^2G=I$. Let $G^\dagger UG$ have the spectral decomposition
    \begin{equation}G^\dagger U G = \sum_{u}\lambda_u\ket{\phi_u}\bra{\phi_u}.\end{equation}
    Then the unitary operator $W=(2GG^\dagger-I)U$ has the spectral decomposition
    \begin{align}
        W &=\sum_{|\lambda_u|=1}\lambda_u\ket{\phi_u^0}\bra{\phi_u^0}\\
        &\quad +\sum_{|\lambda_u|<1}\left(\exp(i\arccos(\lambda_u))\ket{\phi_u^+}\bra{\phi_u^+}\right.\\
        &\quad\quad+\left.\exp(-i\arccos(\lambda_u))\ket{\phi_u^-}\bra{\phi_u^-}\right)
    \end{align}
    where
    \begin{equation}
    \begin{gathered}
        \ket{\phi_u^0}=G\ket{\phi_u}, \ket{\phi_u^1}=\frac{UG\ket{\phi_u}-\lambda_u G\ket{\phi_u} }{\sqrt{1-\lambda_u^2}},\\
        \ket{\phi_u^\pm}=\frac{\ket{\phi_u^0}\pm\ket{\phi_u^1}}{\sqrt{2}}.
    \end{gathered}
    \end{equation}
\end{lemma}
Suppose that a unitary operator $U$ is Hermitian and an $(\alpha, a, 0)$-block-encoding of $A$. Then I can apply \cref{lemma:Q_walk} by taking $G$ to be $I\otimes \ket{0}^{\otimes a}$.

As shown in \cite{LS}, the QSP technique from \cite{LC16} can be applied to the walk operator, as described in the following proposition.
The $\ket{\phi_u^\pm}$ and $\theta_u^\pm$ in the following proposition represent the decomposed eigenstates and their corresponding eigenphases of the walk operator, responsively. The state $\ket{\phi_u^1}$ is orthogonal to $\ket{\phi_u^0}$.

\begin{proposition}[Applying QSP to Walk Operator \cite{LS}]\label{prop:QSP_W}
Suppose $A$ is a Hermitian matrix with the spectral decomposition $A=\sum_{u}\lambda_u\ket{\phi_u}\bra{\phi_u}$. Let $U$ be a Hermitian $(\alpha,a,0)$-block-encoding of $A$, where $\alpha>\|A\|$, and $a\in \mathbb{N}$. Let $(f_A, f_B, f_C, f_D)$ be an achievable polynomial tuple of degree at most $d$. If I define $G:=I\otimes\ket{0}^{\otimes a}$, then the operator 
\begin{align}
&\sum_uOp(f_A(\theta_u^+),f_B(\theta_u^+),f_C(\theta_u^+),f_D(\theta_u^+))\otimes\ket{\phi_u^+}\bra{\phi_u^+}\\
&+\sum_uOp(f_A(\theta_u^-),f_B(\theta_u^-),f_C(\theta_u^-),f_D(\theta_u^-))\otimes\ket{\phi_u^-}\bra{\phi_u^-}
\end{align}
can be implemented using $O(d)$ queries of $U$ and $U^\dagger$. Here, the quantum states $\ket{\phi_u^\pm}$ and angles $\theta_u^\pm$ are defined as
\begin{equation}\label{eq:walk_states_alpha}
\begin{gathered}
    \ket{\phi_u^0}:=G\ket{\phi_u}, \ket{\phi_u^1}:=\frac{UG\ket{\phi_u}-\frac{\lambda_u}{\alpha}G\ket{\phi_u} }{\sqrt{1-\left(\frac{\lambda_u}{\alpha}\right)^2}},\\ \ket{\phi_u^\pm}:=\frac{\ket{\phi_u^0}\pm\ket{\phi_u^1}}{\sqrt{2}}, \theta_u^\pm := \pm\arccos\frac{\lambda_u}{\alpha}.\\
\end{gathered}
\end{equation}
\end{proposition}
The construction of the quantum circuit in \cref{prop:QSP_W}, as described in \cite{LS}, is as follows:

By defining the isometry $G:=I\otimes \ket{0}^{\otimes a}$, I can construct a walk operator $W:=(2GG^\dagger-I)U$. As shown in \cref{lemma:Q_walk}, this walk operator has the following spectral decomposition:
\begin{align}\label{eq:QSP_W}
    W=\sum_{u}\left\{\exp\left(i\arccos\frac{\lambda_u}{\alpha}\right)\ket{\phi_u^+}\bra{\phi_u^+}+\exp\left(-i\arccos\frac{\lambda_u}{\alpha}\right)\ket{\phi_u^-}\bra{\phi_u^-}\right\}.
\end{align}
Given an achievable polynomial tuple $(f_A, f_B, f_C, f_D)$ of degree at most $d$, \cref{lemma:QSP_multi} is then applied to the walk operator $W$, as in \cite{LS}.

\subsection{Variable Time Amplitude Amplification}
Variable-Time Amplitude Amplification (VTAA) generalizes amplitude amplification to setting with non-uniform stopping times. VTAA decomposes the overall algorithm into a series of sub-algorithms, dynamically adjusting the amplification process. In my protocol, VTAA improves the dependence on the condition number from quadratic to nearly linear.

First, I introduce the variable-stopping-time algorithm framework.

\begin{definition}[Variable-stopping-time Algorithm, \cite{ambainis}]
A quantum algorithm $\mathcal{A}$ acting on a Hilbert space $\mathcal{H}$ is a variable-stopping-time algorithm if $\mathcal{A}$ can be decomposed into a sequence of $T$ quantum sub-algorithms, $\mathcal{A}=\mathcal{A}_T\cdot\cdots\cdot\mathcal{A}_1$. This decomposition requires that $\mathcal{A}$ operates on a composite Hilbert space $\mathcal{H}=\mathcal{H}_\mathsf{C}\otimes\mathcal{H}_\mathsf{A}$, where $\mathcal{H}_\mathsf{C}=\bigotimes_{j=1}^T\mathcal{H}_{\mathsf{C}_j}$ with each $\mathcal{H}_{\mathsf{C}_j}=\mathrm{span}(\ket{0},\ket{1})$. Furthermore, each sub-algorithm $\mathcal{A}_j$ acts on $\mathcal{H}_{\mathsf{C}_j}\otimes \mathcal{H}_\mathsf{A}$ and is controlled on the first $j-1$ qubits being in the all-zero state, i.e., $\ket{0}^{\otimes j-1}\in\bigotimes_{i=1}^{j-1}\mathcal{H}_{\mathsf{C}_i}$.
\end{definition}

The following definitions and notations for variable-time amplitude amplification largely follow \cite{CMP}. In variable-time amplitude amplification, the register $\mathsf{A}$ includes a flag register $\mathsf{F}$, such that $\mathcal{H}_\mathsf{F}=\mathrm{Span}(\ket{good},\ket{bad})$ and $\mathcal{H}_\mathsf{A}=\mathcal{H}_\mathsf{F}\otimes \mathcal{H}_\mathsf{W}$.
For $1\leq j \leq T$, let $t_j\quad(t_1<t_2<\cdots<t_T\eqqcolon t_\mathrm{max})$ denote the (query/gate) complexity of the quantum algorithm $\mathcal{A}_j\mathcal{A}_{j-1}\cdots \mathcal{A}_1$.
Furthermore, the probability $p_j$ is defined as:
\begin{equation}
    p_j=\|\Pi_{\mathsf{C}_j}\mathcal{A}_j\mathcal{A}_{j-1}\cdots\mathcal{A}_1\ket{0}_\mathcal{H}\|^2,
\end{equation}
where 
\begin{equation}
    \Pi_{\mathsf{C}_j}=\bigotimes_{k=1}^{j-1}\ket{0}\bra{0}_{\mathsf{C}_k}\otimes\ket{1}\bra{1}_{\mathsf{C}_j}\otimes\bigotimes_{l=j+1}^{T}\ket{0}\bra{0}_{\mathsf{C}_l}\otimes I_{\mathcal{H}_\mathsf{A}}.
\end{equation}
The average stopping time $t_\mathrm{avg}$ is defined as:
\begin{equation}
    t_\mathrm{avg}=\sqrt{\sum_{i=1}^T p_jt_j^2}.
\end{equation}
Finally, the success probability of $\mathcal{A}$ is given by:
\begin{equation}
    p_\mathrm{succ}=\|\Pi_\mathsf{F}\mathcal{A}_T\cdots\mathcal{A}_1\ket{0}_\mathcal{H}\|^2,
\end{equation}
where the projection operator $\Pi_\mathsf{F}$ is defined as
\begin{equation}
    \Pi_\mathsf{F}=I_{\mathcal{H}_\mathsf{C}}\otimes\ket{good}\bra{good}_{\mathcal{H}_\mathsf{F}}\otimes I_{\mathcal{H}_\mathsf{W}}.
\end{equation}

The complexity of VTAA is given by the following theorem.
\begin{theorem}[Variable-time Amplitude Amplification (VTAA), {\cite[Theorem 23]{CGJ}}]\label{theorem:VTAA}
Let $U$ be a state-preparation unitary operator that prepares the state $U\ket{0}^{\otimes k}=\sqrt{p_\mathrm{prep}}\ket{0}\ket{\psi_0}+\sqrt{1-p_\mathrm{prep}}\ket{1}\ket{\psi_1}$ and has query complexity $T_U$. Suppose that $\mathcal{A}$ is a variable-stopping-time algorithm for which I know lower bounds $p_\mathrm{prep}\geq p'_\mathrm{prep}$ and $p_\mathrm{succ}\geq p'_\mathrm{succ}$. Let $t'_\mathrm{max}:=2t_\mathrm{max}/t_1$.
Then, with success probability at least $1-\delta$, I can construct a variable-stopping-time algorithm $\mathcal{A}'$ that prepares a state $q\ket{0}\mathcal{A}'\ket{\psi_0}+\sqrt{1-q^2}\ket{1}\ket{\psi_\mathrm{garbage}}$, such that $q=\Theta(1)$ and $\mathcal{A}'$ has complexity 
\begin{equation}
    O\left(\left(t_\mathrm{max}+\frac{T_U+k}{\sqrt{p_\mathrm{prep}}}\right)\sqrt{\log(t'_\mathrm{max})}+\frac{\left(t_\mathrm{avg}+\frac{T_U+k}{\sqrt{p_\mathrm{prep}}}\right)\log(t'_\mathrm{max})}{\sqrt{p_\mathrm{succ}}}\right).
\end{equation}
\end{theorem}

\section{Quantum Least Squares Protocol}\label{section:quantum_protocol}
    This section describes the quantum protocol for linear regression in \cref{setting:ls}. Since $\mathbf{x}_\mathrm{opt}$ (defined as \cref{eq:x_opt}) is equivalent to $\mathbf{x}_\mathrm{opt}=A^+\mathbf{b}$, I aim to generate a quantum state $\ket{\mathbf{x}_\mathrm{opt}}=A^+\mathbf{b}/\|A^+\mathbf{b}\|$.
The quantum linear regression protocol proposed in \cite{MS} uses the gapped phase estimation algorithm from \cite{CKS}. By introducing the branch-marked gapped phase estimation introduced in \cite{LS}, I develop a quantum protocol that achieves lower quantum communication complexity.

\subsection{The referee constructing some block-encodings}
For \cref{setting:ls}, each party $P_i$ has the matrix $A_i$ for $i \in \{0,\cdots,r-1\}$. To enable quantum computations on the composite matrix $A$, the referee constructs a block-encoding of $A$ defined as \eqref{eq:qcm_of_A_b}.
\begin{proposition}[Constructing Block-encoding of $A$, {\cite[Proposition 14 Step 1.]{MS}}]\label{prop:qc_be_A}
For \cref{setting:ls}, the referee can construct an $(\alpha, \log r+1, 0)$-block-encoding of $A$ with $O(r\log n)$ qubits of communication, where $\alpha = \sqrt{\sum_{i}\|A_i\|^2}$.
\end{proposition}

The gapped phase estimation protocol requires the block-encoding of $\Bar{A}$. \cref{proposition:qc_be_A_Herm} quantifies the quantum communication complexity for constructing this block-encoding.
\begin{proposition}[Constructing Block-encoding of $\Bar{A}$, \cite{MS}]\label{proposition:qc_be_A_Herm}
For \cref{setting:ls}, the referee can construct an $(\alpha, \log r+1, 0)$-block-encoding of $\Bar{A}$ with $O(r\log (mn))$ qubits of communication, where $A$ is defined as \eqref{eq:qcm_of_A_b}, $\alpha = \sqrt{\sum_{i}\|A_i\|^2}$.
\end{proposition}

The quantum communication complexity for the referee to generate the state $\ket{\mathbf{b}}$ is given by the following proposition.
\begin{proposition}[Preparing $\ket{\mathbf{b}}$, \cite{MS}]\label{prop:com_b}
For \cref{setting:ls}, the referee can prepare a quantum state $\ket{\mathbf{b}}$ with $O(r\log n)$ qubits of communication.
\end{proposition}

Matrix inversion constitutes a subprotocol within my variable-time protocol. The following lemma details a known quantum protocol for the referee to implement the matrix inversion operation, approximating $A^+$.
In the following lemma, $\varphi$ represents a threshold parameter: the operation only attempts to invert singular values larger than $\varphi\alpha$. $f(A)$ represents an $\epsilon$-approximation of $A^+$. $\alpha_\mathrm{max}$ is a normalization factor that ensures the output state is properly normalized.

\begin{lemma}[Matrix Inversion, \cite{MS}]\label{lemma:com_mat_inv}
For \cref{setting:ls}, let $A$ be a matrix with the singular value decomposition
\begin{equation}
    A = \sum_{k}\sigma_k\ket{x_k}\bra{y_k}.
\end{equation}
Let $\sigma_\mathrm{min}=\min_k(\sigma_k)$. For any quantum state $\ket{\mathbf{b}}\in \mathrm{Span}(\{\ket{x_k}|\sigma_k/\alpha\geq \varphi\})$, the following operation can be implemented:
\begin{align}
    Inv(\varphi, \varepsilon):\ket{0}_\mathsf{F}\ket{0}_\mathsf{Q}\ket{\mathbf{b}}_\mathsf{I}\mapsto\frac{1}{a_\mathrm{max}}\ket{1}_\mathsf{F}\ket{0}_\mathsf{Q}f(A)\ket{\mathbf{b}}_\mathsf{I}+\ket{0}_\mathsf{F}\ket{\perp}_{\mathsf{QI}}
\end{align}
with
\begin{equation}
O\left(\frac{r}{\varphi}\log (mn)\log \frac{1}{\varphi\varepsilon}\right)
\end{equation}
qubits of communication, where $\alpha = \sqrt{\sum_{i}\|A_i\|^2}$, $a_\mathrm{max}=O(\sigma_\mathrm{min})$ is a constant independent of $\varphi$, $\ket{\perp}_{\mathsf{QI}}$ denotes an unnormalized quantum state orthogonal to $\ket{0}_\mathsf{Q}$ and $f(A)$ satisfies $\|f(A)\ket{\mathbf{b}}-A^+\ket{\mathbf{b}}\|\leq \varepsilon$. 
\end{lemma}

\subsection{Gapped Phase Estimation with Branch Marking}
Gapped Phase Estimation (GPE) is a quantum algorithm designed to identify the spectral interval to which each eigenvalue of an operator belongs. This subsection focuses on the execution of the branch-marked gapped phase estimation developed in \cite{LS} in the quantum coordinator model.
Branch-marked gapped phase estimation uses quantum signal processing. I first present a protocol for implementing QSP through quantum communications, detailed as follows.
\begin{proposition}\label{prop:QSP_com}
For \cref{setting:ls}, let $\Bar{A}$ have the spectral decomposition $\Bar{A}=\sum_{u}\lambda_u\ket{\phi_u}\bra{\phi_u}$. Let $(f_A, f_B, f_C, f_D)$ be an achievable polynomial tuple of degree at most $d$, known to the referee. Let $G:=I\otimes \ket{0}^{\otimes \log r+1}$. Then, the referee can implement
\begin{align}
    &\sum_uOp(f_A(\theta_u^+),f_B(\theta_u^+),f_C(\theta_u^+),f_D(\theta_u^+))\otimes\ket{\phi_u^+}\bra{\phi_u^+}\\
    &+\sum_uOp(f_A(\theta_u^-),f_B(\theta_u^-),f_C(\theta_u^-),f_D(\theta_u^-))\otimes\ket{\phi_u^-}\bra{\phi_u^-}
\end{align}
with $O(rd\log (mn))$ qubits of communication, where $\theta_u^\pm$ and
$\ket{\phi_u^\pm}$ are defined as \cref{eq:walk_states_alpha}.
\end{proposition}
\begin{proof}
The quantum protocol consists of two steps.

Step 1.  The referee constructs the walk operator $W=(2GG^\dagger-I)U$ with quantum communication. Let $U$ be $(\alpha,\log r+1,0)$-block-encoding of $\Bar{A}$. By \cref{proposition:qc_be_A_Herm}, $U$ can be constructed with $O(r\log(mn))$ qubits of communication. Since the referee can construct $G$ without communication, the walk operator $W$ can be constructed with a total of $O(r\log (mn))$ qubits of communication.

Step 2. The referee constructs the desired circuit. Let $U_\phi$ be
\begin{equation}
    U_\phi=R_Z(\phi)_\mathsf{P} H_\mathsf{P}(\ket{0}\bra{0}_\mathsf{P}\otimes I_\mathsf{Q}+\ket{1}\bra{1}_\mathsf{P}\otimes W_\mathsf{Q})H_\mathsf{P}R_Z^\dagger(\phi)_\mathsf{P}.
\end{equation}
The referee can implement $R_Z(\phi)$ and $H$ operations without communication. From Step 1, the construction $W$ requires $O(r\log(mn))$ qubits of communication. Therefore, the referee can implement $U_\phi$ once with $O(r\log(mn))$ qubits of communication. According to \cref{prop:QSP_W}, QSP can be executed with $O(d)$ uses of $U_\phi$. Consequently, the total quantum communication complexity for implementing QSP is $O(rd\log (mn))$ qubits.
\end{proof}

Branch-marked gapped phase estimation relies on branch-marked quantum states. Before introducing the gapped phase estimation algorithm, I first introduce the branch marking algorithm.

According to \cref{eq:QSP_W}, the spectral decomposition of the walk operator $W$ is given by:
\begin{align}
    W&=\sum_{u}\left\{\exp\left(i\arccos\frac{\lambda_u}{\alpha}\right)\ket{\phi_u^+}\bra{\phi_u^+}+\exp\left(-i\arccos\frac{\lambda_u}{\alpha}\right)\ket{\phi_u^-}\bra{\phi_u^-}\right\}.
\end{align}

The states $\ket{\phi_u^\pm}$ represent a division of $\ket{\phi_u}$ into distinct quantum states. These states are labeled with a $+/-$ sign, corresponding to their respective branches. Branch marking is a technique used to record the sign of a state's corresponding branch onto an ancilla qubit.
\begin{theorem}[Branch Marking, {\cite[Proposition 22]{LS}}]\label{theorem:bm}
Let $A$ be a Hermitian matrix that has the spectral decomposition
\begin{equation}
    A = \sum_{u}\lambda_u\ket{\phi_u}\bra{\phi_u}.
\end{equation}
Let $U$ be a unitary operator and $G$ be an isometry such that $G^\dagger UG=A/\alpha$ and 
$G^\dagger U^2G = I$, where $\alpha\geq 2\|A\|$. Let $\ket{\phi_u^\pm}$ be defined as \cref{eq:walk_states_alpha}. For any $\varepsilon>0$, the isometry $BM(\varepsilon)$ satisfying
\begin{equation}\label{eq:branch_marking}
\begin{gathered}
    BM(\varepsilon)\ket{+}\ket{+}\ket{\phi_u^\pm}= \ket{\xi_u^\pm}\ket{\phi_u^\pm}\\
    \|\ket{\xi_u^\pm}-\ket{+}\ket{\pm}\|\leq \varepsilon
\end{gathered}
\end{equation}
can be implemented using $O(\log(1/\varepsilon))$ 
queries to $U$.
\end{theorem}
According to \cite[Proposition 22]{LS}, the branch marking algorithm is constructed as follows. By \cref{lemma:Q_walk}, the walk operator $W=(2GG^\dagger-I)$ has the spectral decomposition
\begin{equation}
    W =\sum_{u}(\exp(i\theta_u^+)\ket{\phi_u^+}\bra{\phi_u^+}+\exp(i\theta_u^-)\ket{\phi_u^-}\bra{\phi_u^-}),
\end{equation}
where $\theta_u^\pm=\pm\arccos(\lambda_u/\alpha)$.
From \cref{lemma:QSP_multi,lemma:BM_APT}, I can implement an operator
\begin{equation}\label{eq:V_BE}      
    V=\sum_u(F(\theta_u^+)\otimes\ket{\phi_u^+}\bra{\phi_u^+}+F(\theta_u^-)\otimes\ket{\phi_u^-}\bra{\phi_u^-})
\end{equation}
such that $\|F(\theta_u^\pm)-(\pm iX)\|\leq \varepsilon$ for all $u$.
The desired isometry $BM(\varepsilon)$ is then implemented by the controlled unitary operator
$\ket{0}\bra{0}\otimes I +\ket{1}\bra{1}\otimes(-iV)$.
This construction is based on \cite{LS}.

Specifically, when $U$ is an $(\alpha,a,0)$-block-encoding of $A$ and $G$ is defined as $G = \ket{0}^{\otimes a} \otimes I$, the assumptions of \cref{theorem:bm} are fulfilled.

Then, I describe the branch-marked gapped phase estimation algorithm. This algorithm determines whether the magnitude of an eigenvalue is larger or smaller than a specific threshold $\varphi$. Here, $\rho$ is a constant that provides the gap around $\varphi$.

\begin{theorem}[Branch Marked Gapped Phase Estimation, {\cite[Proposition 23]{LS}}]\label{theorem:BM_GPE}
Let $A$ be a Hermitian matrix that has the spectral decomposition
\begin{equation}
    A = \sum_{u}\lambda_u\ket{\phi_u}\bra{\phi_u}.
\end{equation}
Let $U$ be a unitary operator and $G$ be an isometry such that $G^\dagger UG=A/\alpha$ and 
$G^\dagger U^2G = I$, where $\alpha\geq 2\|A\|$. Let $\ket{\phi_u^\pm}$ be defined as \cref{eq:walk_states_alpha}. For any $\varepsilon>0, 0<\varphi<1$ and constant $\rho>1$, I consider the isometry $GPE_{BM}$ such that:
\begin{equation}
    GPE_{BM}(\varphi,\varepsilon)\ket{\pm}\ket{0}\ket{\phi_u^\pm}\mapsto\ket{\pm}\ket{\xi_u}\ket{\phi_u^\pm}, 
\end{equation}
where
\begin{equation}
    \begin{cases}
        \|\ket{\xi_u}-\ket{0}\|\leq \varepsilon, & \frac{\lambda_u}{\alpha}\in[\varphi,1)\\
        
        \|\ket{\xi_u}-i\ket{1}\|\leq \varepsilon, & \frac{\lambda_u}{\alpha}\in\left[-\frac{\varphi}{\rho},\frac{\varphi}{\rho}\right]\\
        
        \|\ket{\xi_u}-(-\ket{0})\|\leq \varepsilon, & \frac{\lambda_u}{\alpha}\in(-1,-\varphi]\\
    \end{cases}.
\end{equation}
$GPE_{BM}$ can be implemented using
\begin{equation}
    O\left(\frac{1}{\varphi}\log \frac{1}{\varepsilon}\right)
\end{equation}
queries to $U$.
\end{theorem}
According to \cite[Proposition 23]{LS}, the branch-marked gapped phase estimation algorithm is constructed as follows. From \cref{lemma:QSP_multi,lemma:BM_APT}, I can implement two operators $V_+$ and $V_-$ defined as:
\begin{equation}
V_+=\sum_u(F(\theta_u^+)\otimes\ket{\phi_u^+}\bra{\phi_u^+}),\quad V_-=\sum_u(F(\theta_u^-)\otimes\ket{\phi_u^-}\bra{\phi_u^-})
\end{equation}
such that 
\begin{equation}
    \begin{cases}
        \|F(\theta)-I\|\leq \varepsilon, & \theta\in (0,\arccos(\varphi)]\\
        
        \|F(\theta)-iX\|\leq \varepsilon, & \theta\in \left[\arccos\frac{\varphi}{\rho},\pi-\arccos\frac{\varphi}{\rho}\right]\\
        
        \|F(\theta)-(-I)\|\leq \varepsilon, & \theta\in [\pi-\arccos(\varphi),\pi)\\
    \end{cases}.
\end{equation}
Then, the desired isometry $GPE_{BM}$ can be implemented by a controlled unitary operator
\begin{equation}
\ket{+}\bra{+}\otimes V_+ +\ket{-}\bra{-}\otimes V_-.    
\end{equation}
This construction is based on \cite{LS}.

An unmarking step is needed when I apply the inverse matrix. The complete gapped phase estimation algorithm is as follows.

\begin{proposition}[Gapped Phase Estimation]\label{prop:GPE}
Let $A$ be an operator with the spectral decomposition
\begin{equation}
    A = \sum_{u}\lambda_u\ket{\phi_u}\bra{\phi_u}.
\end{equation}
Let $U$ be an $(\alpha,a,0)$-block-encoding of $A$, where $\alpha\geq 2\|A\|$. For any $\varepsilon>0$ and $0<\varphi<1$, the following transformation can be implemented for all $u$:
\begin{equation}
    GPE(\varphi,\varepsilon)\ket{0}_\mathsf{C}\ket{\phi_u}_\mathsf{I}\mapsto \ket{\xi_u}_\mathsf{C}\ket{\phi_u}_\mathsf{I}, 
\end{equation}
where $\ket{\xi_u}$ satisfies
\begin{equation}
    \begin{cases}
        \|\ket{\xi_u}-\ket{0}\|\leq \varepsilon, & \frac{\lambda_u}{\alpha}\in[\varphi,1)\\
        \|\ket{\xi_u}-i\ket{1}\|\leq \varepsilon, & \frac{\lambda_u}{\alpha}\in[-\frac{\varphi}{2},\frac{\varphi}{2}]\\
        \|\ket{\xi_u}-(-\ket{0})\|\leq \varepsilon, & \frac{\lambda_u}{\alpha}\in(-1,-\varphi]\\
    \end{cases}.
\end{equation}
This transformation can be implemented using
\begin{equation}
    O\left(\frac{1}{\varphi}\log \frac{1}{\varepsilon}\right)
\end{equation}
queries to $U$.
\end{proposition}
\begin{proof}
Let $G:=\ket{0}^{\otimes a}\otimes I$. The implementation of $GPE(\varphi,\varepsilon)$ proceeds in the following steps: 
\begin{enumerate}[Step. 1]
    \item Initialize 2 ancilla qubits in $\ket{+}\ket{+}$ (for branch marking) and $a$ ancilla qubits in $\ket{0}^{\otimes a}$ (for the isometry $G$). The combination of these ancillas with the input state $\ket{\phi_u}$ allows for the implicit transformation of $\ket{\phi_u}$ into a superposition of branch states.
    \item Apply the branch marking $BM(\varepsilon)$ as defined in \cref{theorem:bm}.
    \item Apply the branch-marked gapped phase estimation $GPE_{BM}(\varphi,\varepsilon)$ as defined in \cref{theorem:BM_GPE}, specifically setting the parameter $\rho=2$.
    \item Apply the inverse of the branch marking $BM^{-1}(\varepsilon)$.
    \item Uncompute the $a$ ancilla qubits (corresponding to applying $G^{-1}$) and discard the 2 $\ket{+}\ket{+}$ ancilla qubits, removing all auxiliary registers.
\end{enumerate}

For all $u$, the overall algorithm can be illustrated by the following state evolution:
\begin{align}
    \ket{0}\ket{\phi_u}&\xrightarrow{\text{Step. 1}} \ket{+}\ket{+}\ket{0}\ket{\phi_u}\ket{0}\\
    &= \ket{+}\ket{+}\ket{0}\frac{1}{\sqrt{2}}(\ket{\phi_u^+}+\ket{\phi_u^-})\\
    &\xrightarrow{BM \text{(Step. 2)}} \ket{+}\frac{1}{\sqrt{2}}(\ket{+}\ket{0}\ket{\phi_u^+}+\ket{-}\ket{0}\ket{\phi_u^-})\\
    &\xrightarrow{GPE_{BM} \text{(Step. 3)}} \frac{1}{\sqrt{2}}(\ket{+}\ket{\xi_u}\ket{\phi_u^+}+\ket{-}\ket{\xi_u}\ket{\phi_u^-})\\
    &\xrightarrow{BM^{-1} \text{(Step. 4)}}\ket{+}\ket{+}\ket{\xi_u}\frac{1}{\sqrt{2}}(\ket{\phi_u^+}+\ket{\phi_u^-})\\
    &=\ket{+}\ket{+}\ket{\xi_u}\ket{\phi_u}\ket{0}\\
    &\xrightarrow{\text{Step. 5}}\ket{\xi_u}\ket{\phi_u}.
\end{align}
This entire sequence approximates the desired transformation to within $\varepsilon$ accuracy.
\end{proof}

Since \cref{prop:GPE} relies solely on QSP and unitary operations that the referee can perform locally, this algorithm is executable via quantum communication.

The quantum communication complexity for GPE is given by the following.

\begin{theorem}[GPE for Quantum Coordinator Model]\label{theorem:GPE_com}
For \cref{setting:ls}, the referee can implement $GPE(\varphi,\varepsilon)$ as defined in \cref{prop:GPE} with 
\begin{equation}
O\left(\frac{r}{\varphi}\log (mn)\log \frac{1}{\varepsilon}\right)
\end{equation}
qubits of communication.   
\end{theorem}
\begin{proof}
The implementation of $GPE(\varphi,\varepsilon)$ relies on several sub-protocols: $BM(\varepsilon)$, $GPE_{BM}(\varphi,\varepsilon)$, and operations for attaching and removing ancilla qubits.
\begin{itemize}
\item $BM(\varepsilon)$ uses QSP with an achievable polynomial tuple of degree at most $O(\log (1/\varepsilon))$ (as indicated in \cref{theorem:bm}). From \cref{prop:QSP_com}, implementation of $BM(\varepsilon)$ requires $O(r\log (mn)\log(1/\varepsilon))$ qubits of communication.
\item Similarly, $GPE_{BM}(\varphi,\varepsilon)$ uses QSP with an achievable polynomial tuple of degree at most  $O((1/\varphi)\log (1/\varepsilon))$ (as indicated in \cref{theorem:BM_GPE}). From \cref{prop:QSP_com}, implementation of $GPE_{BM}(\varphi,\varepsilon)$ requires $O((r/\varphi)\log (mn)\log(1/\varepsilon))$ qubits of communication.
\item The referee can attach and remove ancilla qubits without any communication.
\end{itemize}
Since the overall protocol for $GPE(\varphi,\varepsilon)$ consists of sequential applications of these sub-protocols, the referee can execute $GPE(\varphi,\varepsilon)$ with
\begin{equation}
O\left(\frac{r}{\varphi}\log (mn)\log \frac{1}{\varepsilon}\right)
\end{equation}
qubits of communication.   
\end{proof}

\subsection{Quantum Protocol}
The overall quantum protocol for \cref{setting:ls} proceeds as follows. Let $T=\lceil\log (\alpha/\delta)\rceil +1$, where $\alpha=\sqrt{\sum_j\|A_j\|^2}$. I use the VTAA technique with the following registers:
\begin{itemize}
    \item $T$-qubit clock register $\mathsf{C}$: This stores the results of GPE.
    \item 1-qubit clock condition register $\mathsf{K}$: This register's state indicates whether the first $(j-1)$ qubits of the clock register $\mathsf{C}$ are in the $\ket{0}$ state.
    \item 1-qubit flag register $\mathsf{F}$: This indicates the protocol's success.
    \item $\log n$-qubit register $\mathsf{I}$: This register starts as $\ket{\mathbf{b}}$ and ultimately stores the output.
    \item Ancilla registers $\mathsf{Q}$: These are used for block-encoding operations.
\end{itemize}
Initially, the registers $\mathsf{C}$, $\mathsf{K}$, $\mathsf{F}$, and $\mathsf{Q}$ are set to $\ket{0}$.

Let $\varphi_j=2^{-j}$ and $\varepsilon'=\varepsilon/(T\alpha_\mathrm{max})$. For the VTAA algorithm $\mathcal{A}=\mathcal{A}_T\cdots \mathcal{A}_1$, I detail each subalgorithm $\mathcal{A}_j$ as follows: 
\begin{enumerate}[Step 1.]
    \item The referee applies $X$ gate to the register $\mathsf{K}$, conditioned on the first $(j-1)$ qubits of the register $\mathsf{C}$ being in the $\ket{0}_{(\mathsf{C}_1,\cdots,\mathsf{C}_{j-1})}$ state.
    \item Conditioned on the register $\mathsf{K}$ being $\ket{1}_{\mathsf{K}}$ state, apply $GPE(\varphi_j,\varepsilon')$ as defined in \cref{theorem:GPE_com} to the registers $(\mathsf{C}_j, \mathsf{I})$.
    \item The referee applies $X$ gate to the register $\mathsf{K}$, conditioned on the first $(j-1)$ qubits of the register $\mathsf{C}$ being in the $\ket{0}_{(\mathsf{C}_1,\cdots,\mathsf{C}_{j-1})}$ state.
    \item Conditioned on the register $\mathsf{C}_j$ being in the $\ket{1}_{\mathsf{C}_j}$ state, apply $Inv(\varphi_j,T\alpha\varepsilon')$ defined in \cref{lemma:com_mat_inv} to the registers $(\mathsf{F}, \mathsf{Q}, \mathsf{I})$.
\end{enumerate}

The overall quantum communication complexity of the variable-time algorithm above is as follows.

\begin{theorem}\label{theorem:qc-ols}
Suppose $\mathbf{b}$ is spanned by the left singular vectors of $A$. For \cref{setting:ls}, the referee can prepare an $\varepsilon$-approximation of $\ket{\mathbf{x}_\mathrm{opt}}$ with 
\begin{equation}
O\left(\frac{r^{1.5}\|A\|}{\delta}\log(mn)\log\frac{r\|A\|}{\delta\varepsilon}\log\frac{r\|A\|}{\delta}\right)
\end{equation}
qubits of communication, where $\ket{\mathbf{x}_\mathrm{opt}}=\mathbf{x}_\mathrm{opt}/\|\mathbf{x}_\mathrm{opt}\|$, and $\mathbf{x}_\mathrm{opt}$ is defined as \cref{eq:x_opt}.
\end{theorem}
\begin{proof}
Let the singular value decomposition of $A$ be $A=\sum_{k}\sigma_k\ket{x_k}\bra{y_k}$.
From \cref{theorem:GPE_com}, the quantum communication complexity of $GPE(\varphi_j,\varepsilon')$ is 
\begin{equation}
O\left(\frac{r}{\varphi_j}\log (mn)\log \frac{1}{\varepsilon'}\right)=O\left(2^jr\log (mn)\log \frac{1}{\varepsilon'}\right).
\end{equation}
From \cref{lemma:com_mat_inv}, the quantum communication complexity of $Inv(\varphi_j,T\alpha\varepsilon')$ is 
\begin{align}
O\left(\frac{r}{\varphi_j}\log (mn)\log \frac{1}{\varphi_j\varepsilon'}\right)=O\left(2^jr\log (mn)\log \frac{2^j}{\varepsilon'}\right).
\end{align}
Thus, the quantum communication complexity of each subalgorithm $\mathcal{A}_j$ is
\begin{equation}
O\left(2^jr\log (mn)\log \frac{2^j}{\varepsilon'}\right),
\end{equation}
and the quantum communication complexity of the sequence of subalgorithms $\mathcal{A}_j\cdots\mathcal{A}_1$ is
\begin{equation}
t_j=O\left(2^jr\log (mn)\log \frac{\alpha}{\delta\varepsilon'}\right).
\end{equation}
The relevant quantities for the VTAA complexity are as follows: \\
The maximum communication complexity of a single VTAA iteration is
\begin{equation}
    T_\mathrm{max}=O\left(\frac{\alpha r}{\delta}\log (mn)\log \frac{\alpha}{\delta\varepsilon'}\right).
\end{equation}
The average communication complexity is
\begin{equation}
    t_\mathrm{avg}=O\left(r\log(mn)\log\frac{\alpha}{\delta\varepsilon'}\sqrt{\sum_{k}\frac{|\beta_k|^2}{\sigma_k^2}}\right).
\end{equation}
The square root of the success probability is
\begin{align}
    \sqrt{p_\mathrm{succ}}&=\frac{1}{\alpha_\mathrm{max}}\sqrt{\sum_k\frac{|\beta_k|^2}{\sigma_k^2}}+O(T\varepsilon')\\
    &=\Omega\left(\delta\sqrt{\sum_k\frac{|\beta_k|^2}{\sigma_k^2}}\right).
\end{align}
Also,
\begin{equation}
    T'_\mathrm{max}=\frac{2T_\mathrm{max}}{t_1}=O\left(\frac{\alpha}{\delta}\right).
\end{equation}
Let $T_\mathbf{b}$ be the communication cost of preparing the state $\ket{\mathbf{b}}$, which is $O(r\log n)$ from \cref{prop:com_b}.

Finally, from \cref{theorem:VTAA,prop:com_b}, the total quantum communication complexity for the entire VTAA algorithm is
\begin{align}
    &(T_\mathrm{max}+T_\mathbf{b})\sqrt{\log(T'_\mathrm{max})}+\frac{(t_\mathrm{avg}+T_\mathbf{b})\log(T'_\mathrm{max})}{\sqrt{p_\mathrm{succ}}}\\
    &=O\left(\frac{\alpha r}{\delta}\log(mn)\log\frac{\alpha}{\delta\varepsilon}\log\frac{\alpha}{\delta}\right)\\
    &=O\left(\frac{r^{1.5}\|A\|}{\delta}\log(mn)\log\frac{r\|A\|}{\delta\varepsilon}\log\frac{r\|A\|}{\delta}\right).
\end{align}
The last equality follows from $\alpha=\sqrt{\sum_i\|A_i\|}=O(\sqrt{r}\|A||)$.
\end{proof}

The following corollary describes the quantum communication complexity of linear regression for a general input vector $\mathbf{b}$.

\begin{corollary}[Quantum Communication Complexity of Least Squares Regression Protocol]\label{corollary:qc-ols_gamma}
For $\mathbf{x}_\mathrm{opt}$ defined as \cref{eq:x_opt}, let $\ket{\mathbf{x}_\mathrm{opt}}=\mathbf{x}_\mathrm{opt}/\|\mathbf{x}_\mathrm{opt}\|$.
For \cref{setting:ls}, the referee can prepare an $\varepsilon$-approximation of $\ket{\mathbf{x}_\mathrm{opt}}$ with 
\begin{equation}
    O\left(\frac{r^{1.5}\|A\|}{\delta\gamma}\log(mn)\log\frac{r\|A\|}{\delta\varepsilon}\log\frac{r\|A\|}{\delta}\right)
\end{equation}
qubits of communication, where $\gamma=\|\Pi_{\mathrm{col}(A)}\ket{\mathbf{b}}\|$ and $\Pi_{\mathrm{col}(A)}$ is a projection operator representing the projection onto the column space of $A$.
\end{corollary}
\begin{proof}
This is because the success probability $p_\mathrm{succ}$ from the proof of \cref{theorem:qc-ols} is reduced by a factor of $\gamma^2$.
\end{proof}

The comparison of communication complexities is presented in \cref{table:qc_ls}, where $\kappa$ denotes the condition number of the matrix $A$, and $\gamma=\Pi_{\mathrm{col}(A)}\mathbf{b}/\|\mathbf{b}\|$, where $\Pi_{\mathrm{col}(A)}$ represents the projection operator projecting onto the column space of $A$.
As \cite{MS} reports their protocol's quantum communication complexity in $\tilde{O}$ notation (omitting the polylogarithmic factors), I calculated the full complexity for comparison; a brief proof is provided in the \hyperref[section:appendix]{Appendix}. 

\begin{table}[ht]
    \centering
    \begin{tabular}{c|c}
      Protocol & Quantum Communication Complexity\\ \hline
      This paper   & $O\left(\frac{r^{1.5}\|A\|}{\delta\gamma}\log mn \log \frac{r\|A\|}{\delta\varepsilon}\log \frac{r\|A\|}{\delta}\right)$\\
      \cite{MS}   & $O\left(\frac{r^{1.5}\|A\|}{\delta\gamma}\log mn \log^2 \frac{r \|A\|}{\delta\varepsilon}\log \frac{r\|A\|}{\delta}\right)$ \\
      \cite{MS}   & $\Omega(r\kappa)$ \\
    \end{tabular}
    \caption{Quantum communication complexity of least squares}
    \label{table:qc_ls}
\end{table}

Although \cite{LS} presents the tunable VTAA method, it is not adopted in this work. This is because its requirement for a constant multiplicative estimate of $\|A^+\mathbf{b}\|$ incurs an estimation cost that exceeds the cost of preparing $\ket{\mathbf{x}_\mathrm{opt}}=A^+\mathbf{b}/\|A^+\mathbf{b}\|$ itself.

\section{Quantum L2-Regularized Least Squares Protocol}\label{section:l_2-regularization}
This section details my quantum protocol for the $\ell_2$-regularized linear regression problem in the quantum coordinator model. I also analyze its quantum communication complexity. 

In the previous section, I considered the ordinary least-squares method:
\begin{equation}
    \argmin_{\mathbf{x}}\|A\mathbf{x}-\mathbf{b}\|^2
\end{equation}
However, it is susceptible to issues such as overfitting and multicollinearity.
To address these challenges, I introduce $\ell_2$-regularization, which adds a penalty term to the loss function to stabilize the solution. I study the minimization of the $\ell_2$-regularized objective function, $\mathcal{L}_\mathrm{l2}$, given by 
\begin{equation}\label{eq:l2reg}
    \mathcal{L}_\mathrm{l2}(\mathbf{x})=\|A\mathbf{x}-\mathbf{b}\|^2+\lambda\|L\mathbf{x}\|^2,
\end{equation}
where $\lambda>0$ is a hyperparameter and  $L\in\mathbb{R}^{n\times n}$ is a nonsingular matrix.
When $L=I$, this becomes a special case known as the ridge regression, whose objective function $\mathcal{L}_\mathrm{ridge}$ is defined as \cref{eq:ridge_obj}.

In the following, I describe the protocol for performing $\ell_2$-regularized linear regression in \cref{setting:l2reg}.

The solution $\mathbf{x}_\mathrm{l2}$ satisfies
\begin{equation}
    \mathbf{x}_\mathrm{l2}=(A^TA+\lambda L^TL)^{-1}A^T\mathbf{b}.
\end{equation}
This problem can be equivalently reformulated as an ordinary least squares problem. Let me define the augmented matrix
\begin{equation}
    A_L:=\begin{pmatrix}
        A\\
        \sqrt{\lambda} L
    \end{pmatrix}.
\end{equation}
The pseudo-inverse of $A_L$ is given by:
\begin{equation}
    A_L^+
    =\begin{pmatrix}
    (A^TA+\lambda L^TL)^{-1}A^T & (A^TA+\lambda L^TL)^{-1}\sqrt{\lambda}L^T   
    \end{pmatrix}.
\end{equation}
The solution $\mathbf{x}_\mathrm{l2}$ is then obtained as 
\begin{equation}
\mathbf{x}_\mathrm{l2}=\argmin_\mathbf{x}\left\|A_L\mathbf{x}-\begin{pmatrix}
        \mathbf{b}\\
        \mathbf{0}
    \end{pmatrix}\right\|^2.        
\end{equation}
Therefore, if the augmented matrix $A_L$ can be efficiently implemented, the $\ell_2$-regularized least squares problem can be solved using the same quantum protocol developed for ordinary least squares.

The properties of the augmented matrix $A_L$ are summarized in the following lemma:
\begin{lemma}[{\cite[Lemma 31]{CMP}}]
Let a matrix $A\in\mathbb{R}^{m\times n}$ and the positive definite penalty matrix $L\in\mathbb{R}^{n\times n}$ have spectral norms $\|A\|$ and $\|L\|$, respectively. Let $\lambda>0$. Suppose $\delta_L$ is the minimum singular value of $L$.
Then the nonzero minimum singular value of $A_L$, denoted as $\sigma_\mathrm{min}(A_L)$, is lower bounded by
\begin{equation}
    \sigma_\mathrm{min}(A_L)\geq \sqrt{\lambda}\delta_L.
\end{equation}
\end{lemma}

The referee constructs a block encoding of $A_L$ as follows.
\begin{lemma}[Constructing Block-encoding of $A_L$]\label{lemma:BE_A_L}
For \cref{setting:l2reg}, the referee can construct an $(\sqrt{\alpha^2+\lambda\|L\|^2}, \log r+1, 0)$-block-encoding of $A_L$ with $O(r\log n)$ qubits of communication, where $\alpha = \sqrt{\sum_{i}\|A_i\|^2}$.
\end{lemma}
\begin{proof}
I regard the referee itself as a new party with the matrix $\sqrt{\lambda}L$, and then apply \cref{prop:qc_be_A}.
\end{proof}
Therefore, the quantum communication complexity of the $\ell_2$-regularized least squares protocol is given by the following theorem.
\begin{theorem}[Quantum Communication Complexity of $\ell_2$-Regularized Least Squares Protocol]\label{theorem:qc-l2reg}
For \cref{setting:l2reg}, the referee can prepare an $\varepsilon$-approximation of $\ket{\mathbf{x}_\mathrm{l2}}$ with
\begin{equation}
    O\left(\frac{1}{\gamma_\mathrm{l2}}\left(\frac{r^{1.5}\|A\|}{\sqrt{\lambda}\delta_L}+r\kappa_L\right)\log (mn)\log\frac{r\|A\|\kappa_L}{\varepsilon}\log (r\|A\|\kappa_L)\right) 
\end{equation}
qubits of communication. Here, $\kappa_L$ represents the condition number of $L$, $\delta_L$ is the minimum singular value of $L$, and  $\gamma_\mathrm{l2}=(1-\mathcal{L}_\mathrm{l2}(\mathbf{x}_\mathrm{l2})/\|\mathbf{b}\|^2)^{1/2}$, where $\mathbf{x}_\mathrm{l2}$ is defined as \cref{eq:x_l2}.
\end{theorem}
\begin{proof}
    It is proved by considering $\sqrt{\alpha^2+\lambda\|L\|^2}$ and $\sqrt{\lambda}\delta_L$ as $\alpha$ and $\delta$ of \cref{theorem:qc-ols}, respectively. Then the success probability $p_\mathrm{succ}$ from the proof of \cref{theorem:qc-ols} is reduced by a factor of $\gamma_\mathrm{l2}^2$.
\end{proof}

Also, the quantum communication complexity of the ridge regression protocol is as follows.
\begin{corollary}[Quantum Communication Complexity of Ridge Regression Protocol]
For \cref{setting:l2reg}, the referee can prepare an $\varepsilon$-approximation of  $\ket{\mathbf{x}_\mathrm{ridge}}$ with
\begin{equation}
O\left(\frac{1}{\gamma_\mathrm{ridge}}\left(\frac{r^{1.5}\|A\| }{\sqrt{\lambda}}+r\right)\log(mn)\log\frac{r\|A\|}{\lambda\varepsilon}\log\frac{r\|A\|}{\lambda}\right)
\end{equation}
qubits of communication. Here, $\mathbf{x}_\mathrm{ridge}$ is the solution to the ridge regression problem, which minimizes $\mathcal{L}_\mathrm{ridge}(\mathbf{x})$ (\cref{eq:ridge_obj}), and $\ket{\mathbf{x}_\mathrm{ridge}}=\mathbf{x}_\mathrm{ridge}/\|\mathbf{x}_\mathrm{ridge}\|$. Additionally, $\gamma_\mathrm{ridge}:=(1-\mathcal{L}_\mathrm{ridge}(\mathbf{x}_\mathrm{ridge})/\|\mathbf{b}\|^2)^{1/2}$.
\end{corollary}
\begin{proof}
    This is because the singular values of $L$ are all 1.
\end{proof}

The results are presented in \cref{table:qc_l2_reg}, where $\delta_L$ denotes the minimum singular value of the penalty matrix $L$, $\kappa_L$ denotes the condition number of $L$, and $\gamma_\mathrm{l2}:=(1-\mathcal{L}_\mathrm{l2}(\mathbf{x}_\mathrm{l2})/\|\mathbf{b}\|^2)^{1/2}$, $\gamma_\mathrm{ridge}:=(1-\mathcal{L}_\mathrm{ridge}(\mathbf{x}_\mathrm{ridge})/\|\mathbf{b}\|^2)^{1/2}$.

\begin{table*}[ht]
    \centering
    \begin{tabular}{c|c}
      Regression Type & Quantum Communication Complexity\\ \hline
      \begin{tabular}{c} $\ell_2$-regularized \\ least squares \\ regression \end{tabular}   & $O\left(\frac{1}{\gamma_\mathrm{l2}}\left(\frac{r^{1.5}\|A\|}{\sqrt{\lambda}\delta_L}+r\kappa_L\right)\log (mn)\log\frac{r\|A\|\kappa_L}{\varepsilon}\log (r\|A\|\kappa_L)\right) $\\
      Ridge regression  & $O\left(\frac{1}{\gamma_\mathrm{ridge}}\left(\frac{r^{1.5}\|A\| }{\sqrt{\lambda}}+r\right)\log(mn)\log\frac{r\|A\|}{\lambda\varepsilon}\log\frac{r\|A\|}{\lambda}\right)$\\
    \end{tabular}
    \caption{Quantum communication complexity of $\ell_2$-regularized least squares regression and ridge regression}
    \label{table:qc_l2_reg}
\end{table*}

\section{Conclusion}\label{section:conclusion}
In this work, I improved quantum protocols for ordinary least squares regression in the quantum coordinator model. By adapting the branch marking and branch-marked gapped phase estimation \cite{LS} to the quantum coordinator model, I achieved a quadratic improvement in the dependence on precision of the quantum communication complexity for ordinary least squares. 
Also, I extended my protocol to $\ell_2$-regularized least squares regression and analyzed its quantum communication complexity.

\section{Future Work}\label{section:future_work}
One important avenue for future work is closing the gap with lower bounds. According to \cref{table:qc_ls}, my current ordinary least squares result still exhibits a gap compared to existing theoretical lower bounds \cite{MS}. Future research could explore ways to refine the proposed quantum communications to achieve better efficiency or derive tighter lower bounds specifically for the $\ell_2$-regularization problem within the quantum coordinator model.

Another direction is to develop more efficient quantum methods for estimating $\|A^+\mathbf{b}\|$. Currently, the high cost of estimating $\|A^+\mathbf{b}\|$ limits the feasibility of applying techniques such as Tunable Variable-Time Amplitude Amplification \cite{LS} and Block Preconditioning \cite{LS}, which enhance the efficiency of the least squares method. Lowering the estimation cost of $\|A^+\mathbf{b}\|$ could unlock significant advantages in the communication setting for linear regression.

In this paper, I assume that the referee has prior knowledge of $\delta$, the lower bound of the minimum non-zero singular value of $A$. In more realistic scenarios where $\delta$ is not known in advance, a natural direction is to explore quantum protocols for estimating $\delta$ in a distributed setting.

On $\ell_2$-regularized least squares method, the hyperparameter $\lambda$ and penalty matrix $L$ are determined to perform the least squares method. There is room to consider methods to determine this hyperparameter. \cite{YGW} shows a quantum algorithm to determine a hyperparameter on a single party. This might be helpful when considering protocols with multiple parties.

\section*{Acknowledgements}
The author gratefully acknowledges insightful discussions with Dr. Changpeng Shao and Dr. Shantanav Chakraborty.
The author, S. M., would like to take this opportunity to thank the
``THERS Make New Standards Program for the Next Generation Researchers.''

\bibliographystyle{IEEEtran}
\bibliography{sample}

\section*{Appendix}
\phantomsection 
\label{section:appendix}
\addcontentsline{toc}{section}{Appendix}
Reference \cite{MS} reports the quantum communication complexity of their protocol in $\tilde{O}$ notation, omitting the polylogarithmic factors. In this appendix, I derive a more precise bound including these factors.

The protocol described in \cite{MS} uses the Gapped Phase Estimation technique introduced in \cite{CKS}.
For clarity, I refer to this specific version as CKS-GPE.

\begin{lemma}[CKS-GPE {\cite[Lemma 22]{CKS}}]
Let $U$ be a unitary operator and $\ket{\psi}$ be an eigenstate satisfying $U\ket{\psi}=\exp(i\pi\lambda)\ket{\psi}$ for some $\lambda\in[-1,1]$. Let $\varphi\in(0,1/4]$ and $\varepsilon>0$. Then there exists a unitary transformation $GPE(\varphi,\varepsilon)$ that implements
\begin{equation}
    GPE(\varphi,\varepsilon)\ket{0}_\mathsf{C}\ket{0}_\mathsf{P}\ket{\psi}_\mathsf{I}
    =(\beta_0\ket{0}_\mathsf{C}\ket{\gamma_0}_\mathsf{P}+\beta_1\ket{1}_\mathsf{C}\ket{\gamma_1}_\mathsf{P})\ket{\psi}_\mathsf{I}
\end{equation}
such that
\begin{equation}
    \begin{cases}
        |\beta_1|<\varepsilon & \text{if $0\leq \lambda\leq\varphi$}\\
        |\beta_0|<\varepsilon & \text{if $2\varphi\leq \lambda\leq1$}
    \end{cases}
\end{equation}
using $O((1/\varphi)\log (1/\varepsilon))$ queries to $U$.
\end{lemma}
CKS-GPE \cite{CKS} is based on the phase estimation algorithm from \cite[Appendix C]{CEMM} and majority voting procedure.

To execute the phase estimation algorithm, a block-encoding of $\exp(i\pi \Bar{A})$ is required. The following protocol implements this block-encoding.
\begin{lemma}[Hamiltonian Simulation for Quantum Coordinator Model, \cite{MS}]\label{lemma:HS_qcm}
For \cref{setting:ls}, $(1,\log r+3,\varepsilon)$ block-encoding of $\exp(i\pi \Bar{A}/\alpha)$ can be implemented with $O(r\log mn\log (1/\varepsilon))$ qubits of communication.
\end{lemma}

Then, the quantum communication complexity of the GPE protocol according to \cite{MS} is as follows.
\begin{proposition}[Quantum Communication Complexity of the CKS-GPE Protocol in the Quantum Coordinator Model (Refinement of \cite{MS})]
For \cref{setting:ls}, let $\varphi\in(0,1/4]$, $\varepsilon>0$, and $A=\sum_k\sigma_k\ket{x_k}\bra{y_k}$. Then, there exists a quantum protocol that implements the unitary $GPE_{CKS}(\varphi, \varepsilon)$ such that for all $k$:
\begin{equation}
    GPE_{CKS}(\varphi,\varepsilon)\ket{0}_\mathsf{C}\ket{0}_\mathsf{P}\ket{x_k}_\mathsf{I}=\alpha_0\ket{0}_\mathsf{C}\ket{g_0}_\mathsf{P}\ket{x_k}_\mathsf{I}+\alpha_1\ket{1}_\mathsf{C}\ket{g_1}_\mathsf{P}\ket{x_k}_\mathsf{I},
\end{equation}
\begin{equation}
\begin{cases}
        |\alpha_1|<\varepsilon & \text{if  $\sigma_k/\alpha<\varphi$}\\
        |\alpha_0|<\varepsilon & \text{if  $2\varphi<\sigma_k/\alpha$}
    \end{cases}    
\end{equation}
with $O((1/\varphi)\log(1/\varphi\varepsilon)\log(1/\varepsilon))$ qubits of communication. Here $\ket{g_0}, \ket{g_1}$ are $O(\log(1/\varepsilon))$-qubit quantum states.
\end{proposition}
\begin{proof}
I apply the phase estimation algorithm described in \cite[Appendix C]{CEMM} with a target accuracy of $\varphi$.  To ensure this, the block-encoding of $\exp(i\pi \Bar{A})$ must be implemented with an error of $o(\varphi\varepsilon)$. According to \cref{lemma:HS_qcm}, a single execution of the phase estimation protocol requires $O((r/\varphi)\log (mn)\log(1/(\varphi\varepsilon)))$ qubits of communication.

A majority vote is performed to achieve an overall accuracy of $\varepsilon$, which requires $O(\log(1/\varepsilon))$ iterations of the phase estimation protocol.
\end{proof}

In my protocol, the CKS-GPE sub-protocol used in \cite{MS} is replaced with the GPE introduced in \cite{LS}. By applying the matrix inversion protocol and the variable-time algorithm in my protocol, I obtain the quantum communication complexity of the protocol from \cite{MS} as follows.
\begin{proposition}[Logarithmic Order of Quantum Communication Complexity of \cite{MS}'s Protocol]
For \cref{setting:ls}, the quantum communication complexity of \cite{MS}'s protocol is
\begin{equation}
    O\left(\frac{r^{1.5}\|A\|}{\delta\gamma}\log (mn) \log^2 \frac{r \|A\|}{\delta\varepsilon}\log \frac{r\|A\|}{\delta}\right)
\end{equation}
\end{proposition}

\end{document}